\documentclass[letterpaper]{article} 
\usepackage[preprint]{aaai2027}
\usepackage[hyphens]{url}  
\usepackage{graphicx} 
\usepackage{natbib}  
\usepackage{caption} 
\usepackage{booktabs}
\usepackage{multirow}
\usepackage{amsmath,amssymb,amsthm}

\newcommand{\suppref}[1]{Appendix~\ref{#1}}
\newcommand{\qeb}{QEncodeBench}
\newcommand{\Sf}{S_f}
\newcommand{\Shat}{\widehat{S}}
\newcommand{\passk}[1]{\text{pass@}#1}
\DeclareMathOperator{\wcw}{w}

\newtheorem{definition}{Definition}
\newtheorem{proposition}{Proposition}
\newtheorem{remark}{Remark}

\title{\qeb: Can Large Language Models Encode Classical Problems into Verified Quantum Oracles?}
\author{
    Xujun Che\textsuperscript{\rm 1}\corresponding,
    Hanhan Wu\textsuperscript{\rm 2},
    Yuchen Yuan\textsuperscript{\rm 3},
    Chenyang Yu\textsuperscript{\rm 4}
}
\affiliations{
    \textsuperscript{\rm 1}Department of Cybersecurity, University of North Carolina at Charlotte, Charlotte, NC, USA\\
    \textsuperscript{\rm 2}Department of Electrical and Computer Engineering, George Mason University, Fairfax, VA, USA\\
    \textsuperscript{\rm 3}Department of Information Sciences and Technology, George Mason University, Fairfax, VA, USA\\
    \textsuperscript{\rm 4}Department of Computer Science and Engineering, University of North Texas, Denton, TX, USA\\
    xche@charlotte.edu, hwu28@gmu.edu, yyuan21@gmu.edu, chenyangyu@my.unt.edu
}

\begin{document}
\maketitle

\begin{abstract}
Grover search, amplitude amplification, and quantum counting all rely on the same reusable subroutine, a phase oracle, whose construction the algorithms literature takes as given: the classical predicate is assumed to be already encoded as a correct, resource-bounded circuit. We turn this assumption into a measured capability. \qeb{} tasks large language models (LLMs) with encoding classical constraint problems as phase oracles and scores the generated circuits with an adversarially self-validated verifier that decides full solution-set equivalence up to a global phase, with ancillas restored and resource budgets enforced. Sampled basis-state tests, we show, systematically overestimate this ability. Measured this way, models separate sharply: code models without a reasoning mode solve essentially nothing, and enabling native reasoning on identical weights improves accuracy by an order of magnitude. The failures are overwhelmingly semantic rather than syntactic. Two architectures, a unit-verified constraint agent and a neuro-symbolic compilation pipeline, close most of the remaining gap by delegating correctness-critical composition to deterministic procedures. Ablations quantify the contribution of each component, and resource gating exposes an architecture-dependent trade-off between circuit width and depth. Finally, controlled difficulty escalation reveals architecture-specific responses to difficulty structure: different difficulty axes degrade different methods, while the neuro-symbolic pipeline passes every evaluated instance. Code and data are available at \url{https://github.com/chexujun/QEncodeBench}.
\end{abstract}

\section{Introduction}
\label{sec:intro}

Most celebrated quantum speedups are conditional on an input the
literature rarely constructs: the \emph{oracle}. Grover's algorithm
assumes a unitary that flips the phase of exactly the assignments
satisfying a classical predicate, and assumes the predicate can be
evaluated ``in unit time''~\citep{grover1996}; amplitude amplification
and quantum counting reuse the same phase oracle across an entire
algorithm family~\citep{brassard2002amplitude};
pipelines in the style of the Quantum Approximate Optimization
Algorithm (QAOA) likewise begin at constraint
encoding~\citep{farhi2014qaoa}. Turning a concrete problem
instance (\emph{this} SAT formula, \emph{this} graph) into a correct,
resource-bounded oracle circuit is the kind of mechanical but
unforgiving work one would like to delegate to large language
models (LLMs).

Existing evidence suggests this is where most LLMs break, and our
measurements show that even frontier models leave a diagnosable
semantic gap. On the QuantumKatas curriculum benchmark, models reproduce
textbook algorithms reliably yet collapse on constraint encoding
(34.4\% on \textsf{SolveSATWithGrover}, their weakest category),
and 43\% of failures are \emph{logical}: code that runs but
prepares the wrong state~\citep{cruzbenito2026quantumkatas}. Measuring this capability,
however, requires scoring machinery that existing benchmarks do not
provide. Sampled unit tests systematically overestimate correctness:
false-positive rates of 30--62\% have been measured on mainstream code
benchmarks~\citep{li2022alphacode}, and an $80\times$ test
amplification changes model rankings~\citep{liu2023evalplus}.
Distribution-level checks used in quantum testing are in turn provably
blind to phase errors, since $|+\rangle$ and $|-\rangle$ have identical
output distributions~\citep{paltenghi2024survey}. An oracle whose phases are
wrong on a single assignment silently breaks the interference pattern
Grover depends on; in our own measurements, basis-state test suites
certify up to 27.7\% of a model whose true rate is zero.

Oracle encoding has a redeeming property: it is one of the rare
generation tasks where \emph{complete} semantic scoring is feasible.
The target semantics is a finite Boolean function, the artifact a
unitary, and equivalence up to global phase is decidable by
simulation at benchmark scale. \qeb{} exploits this. Our contributions:

\begin{enumerate}
\item \textbf{A benchmark with a self-validated, complete verifier.}
  660 generator-frozen instances, a 480-instance graded core plus
  a 180-instance escalation set, across seven constraint families
  and five difficulty tiers, with isolated seed domains for
  contamination control. Candidates are scored by full solution-set
  equivalence over all $2^n$ assignments plus resource gating, and the verifier validates itself through adversarial mutants,
  cross-method agreement, and live spot checks.
\item \textbf{A measured capability map and a method comparison.}
  Across thirteen model configurations, direct generation ranges
  from zero to 83.1\% semantic pass@1, with the reasoning mode as
  the dominant factor and failures concentrated at the semantic
  level.
  Two architectures, a \emph{unit-verified constraint agent} and a
  \emph{neuro-symbolic compilation} pipeline, reach 93.1\% and
  92.1\% by delegating composition to deterministic code; ablations
  price each component.
\item \textbf{A stress-response matrix.}
  Three structural difficulty axes press on different layers:
  hybrid constraints degrade the unit agent, zero-slack ancilla
  budgets degrade direct generation, scale degrades neither
  scaffold, and neuro-symbolic compilation passes every
  evaluated instance on all three, because the pressured layer is
  exactly the one it delegates to a compiler.
\end{enumerate}

\begin{figure}[t]
\centering
\includegraphics[width=\columnwidth]{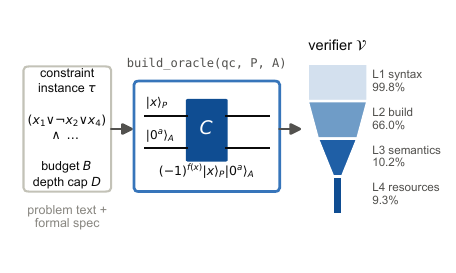}
\caption{The \qeb{} task: encode a constraint instance as a phase
oracle via \texttt{build\_oracle} within qubit budget $B$ and depth
cap $D$, scored by the four-gate verifier.}
\label{fig:task}
\end{figure}

\section{Related Work}
\label{sec:related}

\paragraph{LLMs for quantum code.}
Qiskit HumanEval scores executable Qiskit with hand-written unit
tests, and Qiskit Code Assistant documents the industrial
demand~\citep{vishwakarma2024qiskithumaneval,dupuis2024qiskitassistant}. QCircuitBench covers
quantum algorithm design at scale but scores its oracle suite with
end-to-end test cases and supplies multi-controlled gates to the
model as library files~\citep{yang2025qcircuitbench}. QuantumKatas ports a
fixed curriculum and finds problem encoding
weakest~\citep{cruzbenito2026quantumkatas}, QuanBench scores
process fidelity against one canonical reference circuit,
penalizing different ancilla
layouts~\citep{guo2025quanbench}, and QuanBench+ aligns tasks
across frameworks~\citep{slim2026quanbenchplus}. Threads on quantum
knowledge~\citep{afane2026quantumaudit}, LLMs as circuit
simulators~\citep{wang2025grovergpt}, and circuits as training
data~\citep{apak2024ketgpt} are orthogonal. None of these combines
specification-level complete scoring, resource gating, verifier
self-validation, controlled difficulty scaling, and a systematic
method comparison; that conjunction is what \qeb{} contributes.
Quantum verification tooling, namely equivalence
checking~\citep{burgholzer2020qcec,peham2022zx}, formal
verification~\citep{lewis2022formal}, black-box oracle
testing~\citep{long2025oracle}, and mutation
analysis~\citep{usandizaga2023mutants}, is contrasted with our
verifier's design decisions in \suppref{app:verifier}.

\paragraph{Scoring adequacy and contamination.}
Test-based scoring has measured false-positive rates of 30\% on
HumanEval and 60\% on APPS~\citep{li2022alphacode}, and amplifying
HumanEval's tests $80\times$ drops pass@k by up to 28.9\% and
reorders model rankings~\citep{liu2023evalplus}; SWE-bench's
authors flag execution-test scoring as a
limitation~\citep{jimenez2024swebench}. EvalPlus names the ideal,
formally verifying correctness for arbitrary inputs, and calls it
infeasible in general~\citep{liu2023evalplus}; in our restricted
domain it is feasible, and we implement it. On contamination,
LiveCodeBench documents post-cutoff performance
cliffs~\citep{jain2024livecodebench}, and the dynamic-benchmarking
survey of \citet{chen2025contamination} rates generator-based
designs strongest. \qeb{} is generator-based with three isolated
seed domains, and its frozen-vs-fresh control answers the survey's
central confound, contamination versus complexity.

\paragraph{Method lineage.}
Self-Debugging, Reflexion, and Self-Refine established the
feedback-repair paradigm and, jointly, that gains track the quality
of the feedback signal~\citep{chen2024selfdebug,shinn2023reflexion,
madaan2023selfrefine}. In our domain the model-generated test
signal collapses, since a model that cannot write the oracle
cannot write its checker either (cf.\ the LLM-generated tests of
QAgent,~\citealp{fu2025qagent}, and the ``toxic tests'' of
CodeT,~\citealp{chen2023codet}); we therefore study repair driven
by a trusted verifier and price that signal by ablation. On the
neuro-symbolic side, PAL offloads execution to an
interpreter~\citep{gao2023pal}; Logic-LM adds solver-error feedback
but stagnates because solver errors certify well-formedness, not
correct modeling~\citep{pan2023logiclm}; SatLM argues for
declarative specifications~\citep{ye2023satlm}; LINC, unable to
verify its parses, falls back to majority
voting~\citep{olausson2023linc}. Our specification is
\emph{executable}, so parses are verified exhaustively and repaired
with semantic counterexamples, and our compiler stands on classical
reversible-logic
synthesis~\citep{soeken2018epfl,schmitt2022tweedledum}, whose toolchains take the natural-language-to-specification step
as given. That step is exactly what we measure.

\section{The \qeb{} Benchmark}
\label{sec:benchmark}

\subsection{Task: Phase-Oracle Encoding}
\label{sec:task}

Fix a problem instance with $n$ problem qubits and a predicate
$f:\{0,1\}^n\to\{0,1\}$ with solution set $\Sf=f^{-1}(1)$ and density
$\rho(f)=|\Sf|/2^n$. An instance is a tuple $\tau=(P_\tau, n, B, D)$:
a natural-language problem statement $P_\tau$ (containing the formal
specification), the qubit budget $B\le 20$, and a transpiled-depth cap
$D$. The model must emit Python/Qiskit code defining
\texttt{build\_oracle(qc, problem\_qubits, ancilla\_qubits)}, which
appends gates to \texttt{qc} using the $n$ problem qubits and at most
the $B-n$ supplied ancillas (Figure~\ref{fig:task}).

\begin{definition}[Correct encoding]
\label{def:correct}
Let $C$ be the unitary realized by the candidate on the $q\le B$
qubits it acts on, of which $a=q-n$ are ancillas initialized to
$|0^a\rangle$ (unused supplied ancillas stay idle and are dropped).
$C$ correctly encodes $f$ iff there exists a global phase $\theta$,
independent of $x$, with
\begin{equation}
C\,|x\rangle|0^{a}\rangle \;=\; e^{i\theta}\,(-1)^{f(x)}\,|x\rangle|0^{a}\rangle
\qquad \forall x\in\{0,1\}^n .
\label{eq:oracle}
\end{equation}
\end{definition}

Definition~\ref{def:correct} makes three demands that sampled testing
does not: the phase pattern must be right on \emph{every} assignment;
ancillas must return to $|0^a\rangle$ (residual entanglement silently
destroys downstream interference~\citep{grover1996,che2023trade}); and the circuit
must be a \emph{unitary}: mid-circuit measurement or reset
disqualifies the candidate, since amplitude amplification requires an
invertible, measurement-free operator~\citep{brassard2002amplitude}.
Global-phase equivalence is the standard
notion~\citep[Def.~3.4]{paltenghi2024survey}; in particular a circuit
marking the \emph{complement} of $\Sf$ is accepted, being $e^{i\pi}$
times Eq.~\eqref{eq:oracle}. The equivalence is exact for
uncontrolled uses (Grover search, amplitude amplification); under a
control, as in quantum counting, a global phase becomes a
\emph{relative} one, so deployment there must fix the phase
convention.

Resources are part of the task. $B$ follows a per-family formula
(Table~\ref{tab:families}) and $D$ is four times the transpiled depth
of a reference implementation under a fixed basis
(\texttt{rz,sx,x,cx}, optimization level 1). High-level black boxes (\texttt{UnitaryGate},
\texttt{PhaseOracle}, statevector synthesis) are banned: the task
is to construct the circuit, not look it up.

\subsection{Problem Families and Difficulty Tiers}
\label{sec:families}

\begin{table*}[t]
\centering
\small
\setlength{\tabcolsep}{4pt}
\resizebox{\textwidth}{!}{%
\begin{tabular}{@{}lllll@{}}
\toprule
Family & Predicate & Budget $B$ (flag tiers) & Budget $B$ (counter tiers) & Tiers \\
\midrule
3-SAT & all clauses satisfied & $n+m+1$ & $n+1+\wcw(m)$ & graded, zero-slack, scale \\
3-coloring & all edges bichromatic & $2V+E+1$ & $2V+\wcw(E)+1$ & graded \\
vertex cover & edges covered $\wedge$ $\mathrm{popcount}(x)\le k$ &
  $n+E+\wcw(n)$ & $n+1+\wcw(E)+\wcw(n)$ (zero-slack) & graded, zero-slack, scale \\
subset sum & $\sum_i v_i x_i = t$ &
  \multicolumn{2}{l}{$n+\wcw(\sum_i v_i)$ (all tiers)} & graded, scale \\
Latin square & row/column cells distinct & $2F+P+1$ & $2F+\wcw(P)+2$ & graded \\
string match & pattern occurs at an offset & $n+\mathit{offs}+1$ & $n+1+\wcw(\mathit{offs})$ & graded, zero-slack, scale \\
SAT $\times$ cardinality & SAT $\wedge$ $\mathrm{popcount}(x) \bowtie k$ &
  \multicolumn{2}{l}{$n+m+\wcw(n)$ (all tiers)} & graded \\
\bottomrule
\end{tabular}}
\caption{Problem families. $\wcw(x)=\lceil\log_2(x{+}1)\rceil$; lower
graded tiers budget one flag qubit per constraint, higher graded
tiers a shared counter. Zero-slack sets $B$ to a rolling reference's exact footprint;
scale grows $n$ toward the simulation limit.}
\label{tab:families}
\end{table*}

The seven families cover conjunctive clause logic (3-SAT), encodings with a surjective
color code requiring careful semantics (coloring, Latin squares), cardinality
constraints requiring counters (vertex cover), arithmetic (subset
sum), disjunctive structure (string matching), and a hybrid demanding
\emph{two} techniques in one circuit (SAT $\times$ cardinality).
Solution density is
filtered to $\rho\in[1/64,\,1/2]$ so that verification costs stay
uniform and marked sets are informative. Every instance carries a
reference oracle that must itself pass the full verifier before the
instance is admitted, so ground truth holds by construction,
avoiding the defective-reference problem of HumanEval, where 11\%
of canonical solutions are flawed~\citep{liu2023evalplus}.

Difficulty has two conventional axes (constraint count; tier-wise
budget tightening from per-constraint flags to shared counters) and
three \emph{structural} axes added after the strongest methods
saturated the graded tiers: zero-slack budgets, scale, and the
hybrid family. The graded tiers over the six base families make up
the 480-instance \emph{core set}, which carries all
direct-generation and method experiments, in full or via stratified
subsets; the structural axes make up the 180-instance
\emph{escalation set} (60 zero-slack, 60 scale, 60 hybrid), held back for the stress-response study of
Table~\ref{tab:escalation}.
Together they are the 660 frozen instances.

\subsection{The Four-Level Verifier}
\label{sec:verifier}

The verifier $\mathcal{V}$ evaluates each candidate through four
gates:

\textbf{L1--L2 (syntax, API, sandboxed build).} Parseability plus
a recursive ban scan for oracle-lookup APIs; then the candidate
builds on a fresh circuit in a resource-limited sandbox, with
non-unitary instructions (\texttt{reset}, measurement,
\texttt{initialize}) rejected recursively.

\textbf{L3 (functional semantics).} The core gate decides
Definition~\ref{def:correct} in two stages: a cheap rejection stage
(the \emph{fast path}), and, for every candidate the fast path
accepts, \emph{exhaustive} per-basis-state simulation, at most
$4096$ input states for $n\le12$, whose deterministic verdict is
the verdict of record. Acceptance is therefore a decision
procedure, not a statistical test. The fast path prepares the
uniform superposition (Grover's own initial state,
\citealp{grover1996}), applies $C$, and reads the full phase
pattern from one statevector: with
$|\psi\rangle = C\,(H^{\otimes n}\!\otimes\! I_a)|0^{n+a}\rangle$ and
amplitudes $\psi_x=\langle x,0^a|\psi\rangle$, it checks
(i)~\emph{leakage}: $1-\sum_x|\psi_x|^2\le\varepsilon_\ell$;
(ii)~\emph{magnitudes}: $\big||\psi_x|-2^{-n/2}\big|\le\varepsilon_m$
for all $x$; (iii)~\emph{phases}: $\arg(\psi_x/\psi_{0^n})$ matches
$\pi\,(f(x)\oplus f(0^n))$ within tolerance. Conditions (i)--(iii) are necessary and sufficient \emph{given}
that $C$ acts diagonally on the problem register; to rule out
non-diagonal impostors (basis permutations pass (i)--(iii) on the
uniform state), two \emph{fingerprint} probes with random
per-amplitude magnitudes and phases must be acted on exactly as
that diagonal (\suppref{app:verifier}). The fast path's value is efficiency: for any alignment fixed in
advance, a non-diagonal deviation of size $\eta$ escapes both
fingerprints with probability at most
$(3\cdot2^{n/2}\varepsilon_m/\eta)^2$, i.e.\ at most
$4\times10^{-8}$ for a misrouted basis state
(Proposition~\ref{prop:soundness}, \suppref{app:verifier}; the implemented anchoring is
discussed there).

\noindent The proof (\suppref{app:verifier}) is a small-ball
argument over the fingerprints' independent random phases. Its two
scope caveats, the guarantee being per fixed alignment while the
implemented check anchors the alignment on the probe, and the
assumption that errors are independent of the published fingerprint
seed, are discussed there; both concern this efficiency bound alone,
since no acceptance rests on the fast path.
Every accepted sample behind this paper carries the exhaustive
verdict, and no fast-path acceptance was overturned by the
confirmation (\suppref{app:verifier}). On failures, the verifier
extracts the \emph{marked set} $\Shat(C)$ (assignments carrying the
$\pi$ relative phase) and reports the \emph{mark accuracy}
$\alpha(C)$, the fraction of assignments consistent with the better
of the two admissible global-phase hypotheses
(\suppref{app:verifier}), together with concrete
counterexample assignments: machine-checkable feedback that
repair-based methods consume.

\textbf{L4 (resources).} Transpiled depth $\le D$ under the fixed
basis and $q\le B$.

\paragraph{Scoring convention.}
Throughout, \emph{semantic} $\passk{k}$ denotes survival of
L1--L3 (the capability this paper studies), and \emph{full}
$\passk{k}$ additionally requires L4. We report both: resource compliance is architecture-dependent, and
folding a tunable threshold ($D=4\,d_{\mathrm{ref}}$) into the
primary metric would conflate two distinct failure modes.

\paragraph{Why not off-the-shelf verification?}
Each obvious alternative, equivalence
checking~\citep{burgholzer2020qcec}, ZX-calculus
rewriting~\citep{peham2022zx}, formal
verification~\citep{lewis2022formal}, and statistical
testing~\citep{paltenghi2024survey,long2025oracle}, fails at least
one benchmark requirement; \suppref{app:verifier} gives the
case-by-case contrast. The price of completeness is the scale cap $n\le12$. The return is
scoring whose false-positive risk is zero on the exhaustive path
and quantitatively bounded on the fast path
(Proposition~\ref{prop:soundness}, \suppref{app:verifier}), with
diagnosable counterexamples on every failure.

\paragraph{Verifier self-validation.}
Trust in the scores reduces to trust in $\mathcal{V}$, so
$\mathcal{V}$ is itself under test on four fronts
(\suppref{app:verifier}): 16 adversarial wrong-by-design
mutants, all rejected, with the two basis-permutation mutants
confirming that the fingerprint probes are load-bearing; three-way
agreement of fast path, exhaustive simulation, and an independent
unitary-level check on 500 sampled verdicts; live spot checks that
cross-check hash-selected candidates against exhaustive simulation
in every run and freeze scoring on any disagreement; and a
global-phase regression suite.

\subsection{Frozen Sets and Contamination Control}
\label{sec:frozen}

Instances are drawn from parametric generators with deterministic
seeds; the core set, the fresh regeneration sets, and the escalation
set draw from three disjoint seed domains. In the taxonomy of
\citet{chen2025contamination}, \qeb{} is a rule-based dynamic
benchmark, the class they rate strongest, satisfying their criteria by construction, from verified
references and generators to seed isolation and tiered
parameters. To separate
contamination from difficulty, their central confound, we compare
the frozen set against a \emph{freshly regenerated} 126-instance set
from the untouched seed domain: no-think 10.0\% (frozen, regenerated
run) vs.\ 9.5\% (fresh); reasoning mode 45.4\% vs.\ 52.4\%.
Bootstrap confidence intervals overlap in both cases: no frozen-set advantage is detected at this sample size, and the reasoning configuration in fact scores seven percentage
points \emph{higher} on the fresh set. 

\section{Encoding Methods}
\label{sec:methods}

\begin{figure}[t]
\centering
\includegraphics[width=\columnwidth]{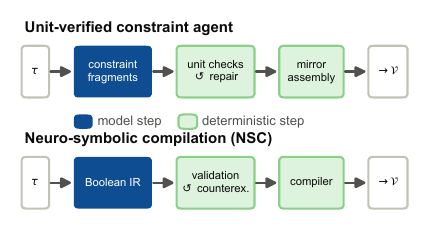}
\caption{Two ways to cut the task at its bottleneck. Both
externalize quantum-implementation correctness; they differ in
where the model stops.}
\label{fig:pipelines}
\end{figure}

All methods receive the same instance text and interface contract,
and all of their outputs pass through the same verifier
$\mathcal{V}$; Figure~\ref{fig:pipelines} contrasts the two
architectures introduced below. ``Thinking'' denotes the model's
native extended reasoning
mode~\citep{wei2022cot}; each method is evaluated with it on and
off.

\subsection{Baselines}
\label{sec:baselines}

\textbf{Direct} prompts for \texttt{build\_oracle} in one shot.
\textbf{Few-shot} adds two worked oracle examples on unrelated toy
predicates. \textbf{Self-repair} feeds back $\mathcal{V}$'s verdict,
counterexamples, and mark accuracy for up to two repair rounds (the
trusted-verifier analogue of
Self-Debugging,~\citealp{chen2024selfdebug}). The \textbf{decomposition agent} first elicits a classical
decomposition, then implements it, testing whether the bottleneck
is classical understanding.

\subsection{Unit-Verified Constraint Agent}
\label{sec:unit}

The unit agent decomposes $f$ \emph{at the quantum implementation
level}. From the instance's formal specification the harness derives
a predicate decomposition
\begin{equation}
f(x)=\textstyle\bigodot_{j=1}^{m} p_j(x),\qquad
\textstyle\bigodot\in\{\bigwedge,\bigvee\},
\label{eq:decomp}
\end{equation}
one $p_j$ per clause, edge, or pattern offset; cardinality
predicates are single units. The model writes one \emph{fragment} per predicate under the
contract
\begin{equation}
F_j\,|x\rangle_{\!P}\,|b\rangle_{\mathrm{fl}}\,|0\rangle_{\mathrm{sc}}
=|x\rangle_{\!P}\,|b\oplus p_j(x)\rangle_{\mathrm{fl}}\,|0\rangle_{\mathrm{sc}},
\label{eq:fragment}
\end{equation}
with no phase side effects. Each fragment is \emph{unit-verified} in isolation on a bench
wired differently from deployment, checking predicate agreement on
all $2^n$ inputs, phase cleanliness, scratch restoration, and a
fingerprint identity test (Eq.~\eqref{eq:fragment} demands a
specific permutation, not correct classical values alone). Failing fragments get up to two rounds of
targeted repair with per-fragment counterexamples, a unit-level form
of process supervision~\citep{lightman2023verify} made exact by the
executable specification. Verified fragments are then composed by a
\emph{deterministic mirror-assembly harness}, which computes all
flags, applies the phase on the aggregate, and uncomputes in reverse,
with three assembly plans chosen by budget, per-flag, rolling
counter, and hold-out rolling (\suppref{app:budgets}). Unlike
decomposed prompting, where the LLM also orchestrates
composition~\citep{khot2023decomposed}, assembly is code, not
generation.

\subsection{Neuro-Symbolic Compilation (NSC)}
\label{sec:nsc}

NSC moves the split one level up: the model never writes quantum
code. It emits a Boolean intermediate representation (IR), a JSON
term over seven primitives
\[
\phi ::= \mathsf{and}(\phi^+) \mid \mathsf{or}(\phi^+) \mid
\mathsf{not}(\phi) \mid \mathsf{bit}(i,v) \mid
\mathsf{cells\_differ}(a,b)
\]
\[
\qquad\;\; \mid\; \mathsf{cell\_ne\_const}(a,v) \mid
\mathsf{linsum\_cmp}(\textstyle\sum_i c_i x_i \bowtie r)
\]
with $\bowtie\in\{=,\ne,\le,\ge\}$. A classical interpreter evaluates
$\phi$ on all $2^n$ assignments and compares against the instance's solution set; mismatches return up to two rounds of concrete
counterexamples: \emph{semantic} feedback, in contrast to the
solver-error feedback whose gains stagnate at well-formedness in
Logic-LM~\citep{pan2023logiclm}. The validated IR is compiled deterministically by one of five
plans (\suppref{app:budgets}), specialized to tight ancilla budgets
and auditability. The compiled
circuit still passes through the full verifier. Models cannot predict program execution~\citep{austin2021mbpp},
and simulating a quantum circuit mentally is strictly harder; NSC
removes every step that requires it. 

\begin{remark}[What NSC does and does not certify]
\label{rem:nsc}
Because instance texts embed the formal specification, a
hand-written parser feeding our compiler would solve the benchmark
deterministically. NSC's score therefore measures the robustness of
the LLM's natural-language-to-IR translation plus the toolchain, an upper reference architecture, while the quantum-implementation
methods remain the benchmark's primary object.
\end{remark}

\section{Experiments}
\label{sec:experiments}

\subsection{Setup}
\label{sec:setup}

\textbf{Models.} Three groups. \emph{Open-weight} (direct
generation only):
Qwen2.5-Coder-7B/32B, Qwen3-32B with reasoning off and on,
DeepSeek-R1-0528-Qwen3-8B, Qwen3-235B-A22B-Thinking, and
gpt-oss-20b/120b at high reasoning effort; one attempt per
instance, non-reasoning models greedy, reasoning models at each model
card's recommended sampling (serving and token-budget protocol in
\suppref{app:results}). \emph{Primary evaluation model},
carrying all agentic methods: DeepSeek-V4-Flash with and without
native reasoning (``thinking''), sampled at $t{=}0.7$ with five
attempts on the core set and three on a stratified 126-instance
subset whose 54-instance core carries the ablations;
Qwen2.5-Coder-7B joins as a weak baseline. \emph{Closed-source models} ($^{\dagger}$/$^{\ddagger}$): Claude
Opus 4.8, Claude Haiku 4.5, and GPT-5.6 Terra
(Table~\ref{tab:main}; \suppref{app:harness}).
\textbf{Metrics.} Unbiased $\passk{k}$~\citep{chen2021codex} with
bootstrap confidence intervals; on the 54-task subsets,
differences below roughly ten percentage points are reported as
suggestive rather than established. Because every verifier acceptance is decided by exhaustive
simulation (Proposition~\ref{prop:soundness}, \suppref{app:verifier}, bounds only how rarely
the fast path wastes a confirmation), $\passk{k}$ doubles as a
deployment number: it is
the realized success rate of verifier-filtered best-of-$k$
selection, the reliable end of a spectrum whose unreliable ends
are model-generated tests~\citep{chen2023codet} and majority
voting~\citep{wang2023selfconsistency}; process-verifier selection
already beats voting~\citep{lightman2023verify}, and ours selects with an exact verifier, which also scores every
row in the paper.

\subsection{Direct-Generation Capability}
\label{sec:e1}

\begin{table*}[t]
\centering
\small
\setlength{\tabcolsep}{4pt}
\begin{tabular}{@{}lcccccccccccccc@{}}
\toprule
 & \multicolumn{4}{c}{Verification funnel (survival, \%)} &
   \multicolumn{6}{c}{By family (L3)} &
   \multicolumn{3}{c}{By graded tier (L3)} \\
\cmidrule(lr){2-5}\cmidrule(lr){6-11}\cmidrule(lr){12-14}
Model / mode & L1 & L2 & L3 & L4 & SAT & Col. & VC & Sum & Latin & Match & 1 & 2 & 3 \\
\midrule
DeepSeek-V4-Flash ($t{=}0.7$) & 99.8 & 66.0 & 10.2 & 9.3 & 8.0 & 2.4 & 1.1 & 9.7 & 19.1 & 24.9 & 15.0 & 10.5 & 4.8 \\
\quad + thinking (greedy) & 85.6 & 73.3 & 45.4 & 40.8 & 51.0 & 25.0 & 12.9 & 62.9 & 65.7 & 61.4 & 59.1 & 51.3 & 25.3 \\
\midrule
Qwen2.5-Coder-7B & 100.0 & 46.9 & 0.0 & 0.0 & 0.0 & 0.0 & 0.0 & 0.0 & 0.0 & 0.0 & 0.0 & 0.0 & 0.0 \\
Qwen2.5-Coder-32B & 100.0 & 43.1 & 0.2 & 0.2 & 0.0 & 0.0 & 1.4 & 0.0 & 0.0 & 0.0 & 0.6 & 0.0 & 0.0 \\
DeepSeek-R1-0528-Qwen3-8B & 86.0 & 16.7 & 1.0 & 0.8 & 0.0 & 0.0 & 0.0 & 1.4 & 1.4 & 4.3 & 3.0 & 0.0 & 0.0 \\
Qwen3-32B & 99.8 & 44.6 & 1.9 & 1.9 & 1.0 & 0.0 & 1.4 & 0.0 & 0.0 & 10.0 & 4.3 & 1.3 & 0.0 \\
\quad + thinking & 99.8 & 58.3 & 20.4 & 19.8 & 7.0 & 0.0 & 2.9 & 70.0 & 37.1 & 20.0 & 32.9 & 19.6 & 8.2 \\
Qwen3-235B-A22B-Thinking & 99.8 & 56.9 & 34.2 & 31.9 & 33.0 & 11.0 & 20.0 & 65.7 & 38.6 & 47.1 & 49.4 & 34.8 & 17.7 \\
gpt-oss-20b (high) & 90.0 & 74.4 & 68.8 & 60.4 & 75.0 & 43.0 & 58.6 & 85.7 & 82.9 & 75.7 & 82.3 & 74.1 & 49.4 \\
gpt-oss-120b (high) & 97.7 & 94.2 & 81.2 & 69.0 & 92.0 & 60.0 & 70.0 & 90.0 & 90.0 & 90.0 & 90.9 & 87.3 & 65.2 \\
\midrule
Claude Haiku 4.5$^{\dagger}$ & 100.0 & 79.4 & 49.0 & 45.2 & 46.0 & 17.0 & 38.6 & 57.1 & 81.4 & 68.6 & 65.9 & 50.0 & 30.4 \\
Claude Opus 4.8$^{\dagger}$ & 99.8 & 90.2 & 65.2 & 59.0 & 61.0 & 49.0 & 67.1 & 58.6 & 77.1 & 87.1 & 72.6 & 64.6 & 58.2 \\
GPT-5.6 Terra$^{\ddagger}$ & 99.4 & 95.8 & 83.1 & 80.8 & 82.0 & 79.0 & 78.6 & 87.1 & 85.7 & 88.6 & 82.3 & 88.0 & 79.1 \\
\bottomrule
\end{tabular}
\caption{Direct generation on the 480-instance core set. The funnel
columns give cumulative survival through the four verifier gates:
L3 is the semantic $\passk{1}$ used as the
primary metric throughout, L4 additionally enforces resource
budgets. Families: 3-SAT, 3-coloring,
vertex cover, subset sum, Latin square, string matching (hybrid
is escalation-only, Table~\ref{tab:escalation}). Blocks, top to bottom: the primary evaluation model (five samples
at $t{=}0.7$ in the sampled row), open-weight models (reasoning
rows sample once at model-card settings), and closed-source models
($^{\dagger}$/$^{\ddagger}$, one attempt each); all remaining rows
decode greedily once. $^{\dagger}$/$^{\ddagger}$Run under non-standard protocols; see
\suppref{app:harness}.}
\label{tab:main}
\end{table*}

Table~\ref{tab:main} maps a wide capability range. At the floor, non-reasoning coders solve essentially nothing
regardless of scale: the 7B and 32B coders land at 0.0 and 0.2\%
while writing syntactically near-perfect Qiskit and building
roughly half the time. The failure is semantic, not syntactic. Reasoning is the first-order
variable: on \emph{identical} Qwen3-32B weights, switching to the
reasoning configuration lifts L3 from 1.9\% to 20.4\%, and native
reasoning quadruples DeepSeek-V4-Flash from 10.2\% (32.7\% across
five verified attempts) to 45.4\%. The contrast survives matched decoding, rising from
10.8\% to 48.9\% under identical $t{=}0.7$ sampling
(Table~\ref{tab:ladder}). Scale then compounds it, 20.4 to 34.2 from Qwen3-32B to 235B, and
the strongest open models, gpt-oss-20b and 120b, reach 68.8 and
81.2\%.
The per-family columns rank difficulty
essentially identically across the DeepSeek and open-weight rows:
vertex cover and coloring are hardest, at 70.0 and 60.0 even for
gpt-oss-120b, and these are exactly the two families whose
encodings hinge on dedicated techniques, a cardinality counter and
the surjective color code, while string
matching, Latin squares, and subset sum lead. The Claude rows land
between the Qwen and gpt-oss blocks at 65.2 and 49.0, keep coloring hardest, and deviate only mildly from the family
pattern. The GPT-5.6 Terra row tops the table at 83.1 semantic and keeps
the same family ordering, with coloring and vertex cover lowest
(\suppref{app:harness}). Per-tier
rates fall
monotonically in every configuration except that proxy row
(e.g.\ 90.9$\to$65.2 for
gpt-oss-120b, 59.1$\to$25.3 for DeepSeek-V4-Flash thinking); the
graded tiers grade across the full capability range. The funnel columns locate each row's difficulty. Coder models fail
at the build and semantics gates, DeepSeek-V4-Flash loses 56 of its
66 surviving points at L3, and the gpt-oss models leave the largest
mass at the resource gate: 8.4 and 12.2 points of transpiled-depth
overruns (\suppref{app:results}).
Failures still
concentrate at exactly the level
sampled tests measure poorly. The
frozen-vs-fresh control shows no memorization advantage.

\subsection{How Much Do Sampled Tests Overestimate?}
\label{sec:scoring}

Re-scoring every direct-generation sample under two degraded
regimes, random basis-state input--output checks (the analogue of
prior benchmarks' end-to-end
suites~\citep{yang2025qcircuitbench}) and four random
product-state probes, quantifies what complete scoring adds
(protocol and per-row table in \suppref{app:results}).

Basis-state suites certify 27.7\% of a model whose true semantic
rate is \emph{zero} (all 133 Qwen circuits they accept are wrong),
and inflate the commodity rows by 25 and 21 points, with 71\% and
31\%
of accepted circuits wrong. Quadrupling the suite from 16 to 64
tests changes nothing: a phase oracle acts as
the identity on every basis state~\citep{paltenghi2024survey}, so
the blindness is structural, unlike the classical setting where test
amplification recovers much of the gap~\citep{liu2023evalplus}. Four
random superposition probes already cut the false-positive rate to
at most 1.2\%; complete verification closes the last gap exactly.
Rankings across the three rows happen to survive every regime (the
gaps are large), but the absolute scores that benchmarks
report do not.

\subsection{Method Comparison}
\label{sec:ladder}

\begin{table}[t]
\centering
\small
\setlength{\tabcolsep}{3.5pt}
\begin{tabular}{@{}lcccc@{}}
\toprule
 & \multicolumn{2}{c}{$\passk{1}$} & \multicolumn{2}{c}{$\passk{3}$} \\
\cmidrule(lr){2-3}\cmidrule(lr){4-5}
Method & L3 & L4 & L3 & L4 \\
\midrule
Direct & 10.8 & 9.8 & 23.7 & 22.2 \\
\quad + thinking & 48.9 & 45.0 & 73.0 & 64.3 \\
\addlinespace
Few-shot & 10.0 & 9.7 & 20.9 & 19.9 \\
\quad + thinking & 52.6 & 49.2 & 77.8 & 68.3 \\
\addlinespace
Self-repair & 27.8 & 27.0
  & \multicolumn{2}{c}{\multirow{2}{*}{($k{=}1$)}} \\
\quad + thinking & 59.5 & 59.5 & \multicolumn{2}{c}{} \\
\addlinespace
Decomposition & 27.8 & 27.8
  & \multicolumn{2}{c}{\multirow{2}{*}{($k{=}1$)}} \\
\quad + thinking & 59.5 & 58.7 & \multicolumn{2}{c}{} \\
\addlinespace
Unit agent & 36.0 & 34.7 & 51.6 & 48.4 \\
\quad + thinking & \textbf{93.1} & 80.4 & \textbf{100.0} & 91.3 \\
\addlinespace
NSC$^{\ast}$ & 92.1 & \textbf{91.8} & 97.6 & \textbf{97.6} \\
\bottomrule
\end{tabular}
\caption{Method comparison (DeepSeek-V4-Flash, stratified 126-task
subset, seven per family$\times$tier cell, \%); ``+ thinking'' switches on the model's reasoning mode, as in
Table~\ref{tab:main}, with L3/L4 as there. $\passk{k}$ is the unbiased estimate from the available samples
($k{=}5$ for the non-thinking direct and few-shot rows, $k{=}3$
elsewhere); ($k{=}1$) marks the single-attempt rows.
$^{\ast}$NSC has no thinking variant by design: it asks the model only
for IR translation. Bold marks the best value
in each column.}
\label{tab:ladder}
\end{table}

Table~\ref{tab:ladder} contains the paper's second set of results.
The non-thinking rows preserve the same method ordering at a
lower level, so the architecture effects are not an artifact of the
reasoning mode; NSC is the exception in kind, reaching 92.1\%
\emph{on the non-reasoning model}, since the only step it asks of
the model is IR translation. Three further observations. First, \emph{architecture dominates additional sampling at
model-token cost}:
two further samples lift direct generation by 24 points, from 48.9
to 73.0, while a single unit-agent attempt is worth 44, and its
$\passk{3}$ then solves the subset entirely. Figure~\ref{fig:frontier} lines the methods up with their
sampling trajectories and per-attempt token cost. Second,
\emph{where you cut matters more than whether you cut}: the
decomposition agent, which decomposes the \emph{classical}
understanding, lands exactly on
plain self-repair (59.5 vs 59.5); the unit
agent, which decomposes the \emph{quantum implementation}, gains some 34 points over both. Classical understanding is not the bottleneck:
NSC's no-validation ablation already reaches 81.5\% end-to-end
(Table~\ref{tab:ablation}), so first-attempt IR translations are largely correct. This sharpens, for our domain, the task-dependence of
least-to-most prompting~\citep{zhou2023leasttomost}.
Third, \emph{resource compliance separates the winners}: at the L4
gate, NSC holds 91.8\% while the unit agent drops to 80.4\%,
because its rolling assembly buys ancilla thrift at a depth premium (\suppref{app:results}).
Architectures trade circuit width for depth; the resource gate
makes the trade visible. The separation is threshold-robust: NSC stays ahead of the unit
agent at every cap we sweep, by 27.7 points at the tightest and 2.4
at the loosest, so the ordering is stable while the magnitude
belongs to the cap (\suppref{app:results}).

On Qwen2.5-Coder-7B, self-repair and the decomposition agent both
score 0/54: this class
of scaffolding has a code-reliability floor below which it cannot
help.

\begin{figure}[t]
\centering
\includegraphics[width=\columnwidth]{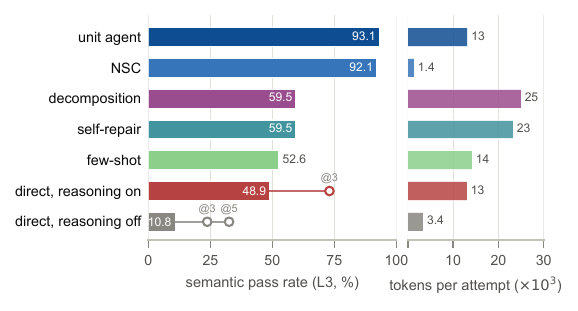}
\caption{Buying samples vs.\ buying architecture (126-task
stratified subset). Bars, rows aligned: each method's $\passk{1}$ (left) and its
median tokens per attempt (right). Hollow circles: the sampling
trajectories, unbiased $\passk{3}$ and $\passk{5}$. Attempt cost
sums all model calls; verifier and compiler compute are not
charged.}
\label{fig:frontier}
\end{figure}

\subsection{Component Ablations}
\label{sec:ablation}

\begin{table}[t]
\centering
\small
\setlength{\tabcolsep}{5pt}
\begin{tabular}{@{}llcc@{}}
\toprule
Component & Ablation & $\passk{1}$ & $\passk{3}$ \\
\midrule
\multirow{3}{*}{Repair rounds}
 & 0 rounds & 61.1 & 88.9 \\
 & 1 round & 80.2 & 98.1 \\
 & 2 (full) & 91.4 & 100.0 \\
\addlinespace
Feedback detail & which-only & 85.8 & 98.1 \\
Mirror assembly & self-assembly & 25.3 & 35.2 \\
\multirow{2}{*}{NSC IR validation}
 & full (main) & 91.4 & 96.3 \\
 & off & 81.5 & 83.3 \\
\bottomrule
\end{tabular}
\caption{Unit-agent and NSC ablations (54 tasks, L3 gate, \%;
unit-agent rows use the reasoning mode, NSC rows the
non-reasoning model).
``Which-only'' reports only which fragment failed;
``self-assembly'' lets the model compose its own verified fragments.}
\label{tab:ablation}
\end{table}

Table~\ref{tab:ablation} prices each ingredient. \textbf{Repair dose} is monotone and unsaturated, rising 61.1 to 80.2 to 91.4 over two repair rounds.
\textbf{Feedback granularity}: telling the model only \emph{which} fragment failed already
recovers most of the gain, and concrete counterexamples add a
suggestive 5.6 points more. Localization is the payload. \textbf{Deterministic assembly is the single largest
component}: letting the model assemble its own verified fragments
collapses pass@1 to 25.3\%, \emph{below the 48.1\% direct
baseline on the same 54-task subset}.
Verified parts composed by unverified generation are a net loss, the
sharpest evidence for taking orchestration away from the
model~\citep{khot2023decomposed}. \textbf{Specification supervision} is worth 9.9 points: without interpreter validation NSC still reaches 81.5\%,
bracketing deployment without a machine-checkable specification
between 81.5 and 91.4\%. The
feedback uses only information mechanically derivable from the public
instance text, but it must be priced, and we price it. The broader pattern,
external verified signals beating internal self-feedback,
replicates in a domain where the model cannot even run its own
output~\citep{madaan2023selfrefine,shinn2023reflexion}.

\subsection{Difficulty Escalation and the Stress-Response Matrix}
\label{sec:escalation}

After the strongest methods saturated the graded tiers, we
escalated difficulty \emph{structurally} on the escalation set,
never by relaxing verification.

\begin{table}[t]
\centering
\small
\setlength{\tabcolsep}{5pt}
\begin{tabular}{@{}llcccc@{}}
\toprule
Method & & Overall & Zero-slack & Scale & Hybrid \\
\midrule
\multirow{2}{*}{Direct} & $\passk{1}$ & 10.5 & 9.3 & 20.4 & 1.9 \\
 & $\passk{3}$ & 24.1 & 22.2 & 44.4 & 5.6 \\
\addlinespace
\multirow{2}{*}{Unit agent} & $\passk{1}$ & 64.2 & 79.6 & 79.6 & 33.3 \\
 & $\passk{3}$ & 83.3 & 100 & 88.9 & 61.1 \\
\addlinespace
\multirow{2}{*}{NSC} & $\passk{1}$ & 100 & 100 & 100 & 100 \\
 & $\passk{3}$ & 100 & 100 & 100 & 100 \\
\bottomrule
\end{tabular}
\caption{Escalation results (L3, \%; stratified 54-task escalation
subset, 18 per axis, $\times3$; thinking for direct and unit,
non-reasoning NSC; L4 within two points). Cells are task-level binomials
at best-of-3; intervals Clopper--Pearson, contrasts Fisher exact.}
\label{tab:escalation}
\end{table}

Table~\ref{tab:escalation}
reads as a matrix rather than a ranking: hybrid structure is the
unit agent's true weak point. $\passk{3}$ drops from 100\% to 61.1\%, 11 of 18 tasks against 34 of 36 elsewhere (Fisher exact
$p{=}0.004$). Zero-slack budgets sharply degrade direct generation, to 9.3\%,
yet leave the unit agent untouched at 100\%, four of 18 tasks
against 18 of 18 ($p{=}1.6{\times}10^{-6}$), because its assembly
harness absorbs that planning burden.
Scale alone degrades neither scaffolded method. And no NSC failure is observed anywhere: 54 of 54 tasks pass, with exact 95\% lower bounds 81.5\% per axis, 93.4\% pooled. Every axis presses on the quantum
implementation layer, and that is the layer NSC delegates to its compiler (Remark~\ref{rem:nsc}); pressuring
semantic \emph{understanding} instead requires obfuscated or
implicit specifications, a natural future escalation axis.

A ceiling absorbed by a symbolic layer is a feature with a warning
label: benchmarks should state which layer their scores certify;
ours do.

\section{Limitations}
\label{sec:limitations}

Simulation-complete scoring caps $n$ at 12 and total qubits at 20,
so conclusions concern encoding capability, not asymptotic scale.
Strong reasoning models compress the graded-tier headroom,
answered by the re-derivable escalation axes. Method comparison,
ablations, and escalation all run on one commodity model family,
and all circuits are Qiskit, with multi-framework transfer
untested~\citep{slim2026quanbenchplus}.

\section{Conclusion}
\label{sec:conclusion}

\qeb{} turns a standing assumption of quantum algorithmics, that a correct phase oracle is simply given, into a
measured and completely verified LLM capability. Three findings organize the results: reasoning, not scale, first moves direct generation; samples help far less than architecture; and delegating correctness-critical
composition to deterministic machinery closes the gap at the
system level. The core-plus-escalation design, graded tiers that map capability and structural axes that
separate what each architecture is good at,
may be a useful pattern as scaffolded systems overtake raw models. Instances, generators, verifier, and scaffolds are publicly released.

\bibliography{references}

@article{vishwakarma2024qiskithumaneval,
  title   = {Qiskit {HumanEval}: An Evaluation Benchmark for Quantum Code Generative Models},
  author  = {Vishwakarma, Sanjay and Harkins, Francis and Golecha, Siddharth and Bajpe, Vishal Sharathchandra and Dupuis, Nicolas and Buratti, Luca and Kremer, David and Faro, Ismael and Puri, Ruchir and Cruz-Benito, Juan},
  journal = {arXiv preprint arXiv:2406.14712},
  year    = {2024}
}

@article{dupuis2024qiskitassistant,
  title   = {Qiskit Code Assistant: Training {LLMs} for Generating Quantum Computing Code},
  author  = {Dupuis, Nicolas and Buratti, Luca and Vishwakarma, Sanjay and Forrat, Aitana Viudes and Kremer, David and Faro, Ismael and Puri, Ruchir and Cruz-Benito, Juan},
  journal = {arXiv preprint arXiv:2405.19495},
  year    = {2024}
}

@inproceedings{yang2025qcircuitbench,
  title     = {{QCircuitBench}: A Large-Scale Dataset for Benchmarking Quantum Algorithm Design},
  author    = {Yang, Rui and Wang, Ziruo and Gu, Yuntian and Chen, Tianyi and Liang, Yitao and Li, Tongyang},
  booktitle = {Advances in Neural Information Processing Systems (NeurIPS)},
  year      = {2025},
  note      = {arXiv:2410.07961}
}

@inproceedings{guo2025quanbench,
  title     = {{QuanBench}: Benchmarking Quantum Code Generation with Large Language Models},
  author    = {Guo, Xiaoyu and Wang, Minggu and Zhao, Jianjun},
  booktitle = {Proceedings of the 40th IEEE/ACM International Conference on Automated Software Engineering (ASE)},
  year      = {2025},
  note      = {arXiv:2510.16779}
}

@inproceedings{slim2026quanbenchplus,
  title     = {{QuanBench+}: A Unified Multi-Framework Benchmark for {LLM}-Based Quantum Code Generation},
  author    = {Slim, Ali and Hamieh, Haydar and Kotaich, Jawad and Ghosn, Yehya and Chehimi, Mahdi and Mohanna, Ammar and Hammoud, Hasan Abed Al Kader and Ghanem, Bernard},
  booktitle = {ICLR Workshop},
  year      = {2026},
  note      = {arXiv:2604.08570}
}

@article{cruzbenito2026quantumkatas,
  title   = {Qiskit {QuantumKatas}: Adapting {Microsoft}'s Quantum Computing Exercises for {LLM} Evaluation},
  author  = {Cruz-Benito, Juan and Faro, Ismael},
  journal = {arXiv preprint arXiv:2605.27210},
  year    = {2026}
}

@article{afane2026quantumaudit,
  title   = {Quantum-Audit: Evaluating the Reasoning Limits of {LLMs} on Quantum Computing},
  author  = {Afane, Mohamed and Laufer, Kayla and Wei, Wenqi and Mao, Ying and Farooq, Junaid and Wang, Ying and Chen, Juntao},
  journal = {arXiv preprint arXiv:2602.10092},
  year    = {2026}
}

@article{fu2025qagent,
  title   = {{QAgent}: An {LLM}-based Multi-Agent System for Autonomous {OpenQASM} Programming},
  author  = {Fu, Zhenxiao and Jiang, Lei and Xu, Yilun and Huang, Gang and Chen, Fan},
  journal = {arXiv preprint arXiv:2508.20134},
  year    = {2025}
}

@article{wang2025grovergpt,
  title   = {{GroverGPT}: A Large Language Model with 8 Billion Parameters for Quantum Searching},
  author  = {Wang, Haoran and Li, Pingzhi and Chen, Min and Cheng, Jinglei and Liu, Junyu and Chen, Tianlong},
  journal = {arXiv preprint arXiv:2501.00135},
  year    = {2025}
}

@inproceedings{apak2024ketgpt,
  title     = {{KetGPT}: Dataset Augmentation of Quantum Circuits Using Transformers},
  author    = {Apak, Boran and Bandic, Medina and Sarkar, Aritra and Feld, Sebastian},
  booktitle = {International Conference on Computational Science (ICCS)},
  year      = {2024},
  note      = {arXiv:2402.13352}
}

@article{chen2021codex,
  title   = {Evaluating Large Language Models Trained on Code},
  author  = {Chen, Mark and Tworek, Jerry and Jun, Heewoo and Yuan, Qiming and Pinto, Henrique Ponde de Oliveira and Kaplan, Jared and Edwards, Harri and Burda, Yuri and Joseph, Nicholas and Brockman, Greg and others},
  journal = {arXiv preprint arXiv:2107.03374},
  year    = {2021}
}

@article{austin2021mbpp,
  title   = {Program Synthesis with Large Language Models},
  author  = {Austin, Jacob and Odena, Augustus and Nye, Maxwell and Bosma, Maarten and Michalewski, Henryk and Dohan, David and Jiang, Ellen and Cai, Carrie and Terry, Michael and Le, Quoc and Sutton, Charles},
  journal = {arXiv preprint arXiv:2108.07732},
  year    = {2021}
}

@inproceedings{liu2023evalplus,
  title     = {Is Your Code Generated by {ChatGPT} Really Correct? Rigorous Evaluation of Large Language Models for Code Generation},
  author    = {Liu, Jiawei and Xia, Chunqiu Steven and Wang, Yuyao and Zhang, Lingming},
  booktitle = {Advances in Neural Information Processing Systems (NeurIPS)},
  year      = {2023},
  note      = {arXiv:2305.01210}
}

@article{li2022alphacode,
  title   = {Competition-Level Code Generation with {AlphaCode}},
  author  = {Li, Yujia and Choi, David and Chung, Junyoung and Kushman, Nate and Schrittwieser, Julian and Leblond, R{\'e}mi and Eccles, Tom and Keeling, James and Gimeno, Felix and Dal Lago, Agustin and others},
  journal = {Science},
  volume  = {378},
  number  = {6624},
  pages   = {1092--1097},
  year    = {2022}
}

@inproceedings{jimenez2024swebench,
  title     = {{SWE-bench}: Can Language Models Resolve Real-World {GitHub} Issues?},
  author    = {Jimenez, Carlos E. and Yang, John and Wettig, Alexander and Yao, Shunyu and Pei, Kexin and Press, Ofir and Narasimhan, Karthik},
  booktitle = {International Conference on Learning Representations (ICLR)},
  year      = {2024},
  note      = {arXiv:2310.06770}
}

@article{jain2024livecodebench,
  title   = {{LiveCodeBench}: Holistic and Contamination Free Evaluation of Large Language Models for Code},
  author  = {Jain, Naman and Han, King and Gu, Alex and Li, Wen-Ding and Yan, Fanjia and Zhang, Tianjun and Wang, Sida and Solar-Lezama, Armando and Sen, Koushik and Stoica, Ion},
  journal = {arXiv preprint arXiv:2403.07974},
  year    = {2024}
}

@article{chen2025contamination,
  title   = {Recent Advances in Large Language Model Benchmarks against Data Contamination: From Static to Dynamic Evaluation},
  author  = {Chen, Simin and Chen, Yiming and Li, Zexin and Jiang, Yifan and Wan, Zhongwei and He, Yixin and Ran, Dezhi and Gu, Tianle and Li, Haizhou and Xie, Tao and Ray, Baishakhi},
  journal = {arXiv preprint arXiv:2502.17521},
  year    = {2025}
}

@inproceedings{chen2023codet,
  title     = {{CodeT}: Code Generation with Generated Tests},
  author    = {Chen, Bei and Zhang, Fengji and Nguyen, Anh and Zan, Daoguang and Lin, Zeqi and Lou, Jian-Guang and Chen, Weizhu},
  booktitle = {International Conference on Learning Representations (ICLR)},
  year      = {2023},
  note      = {arXiv:2207.10397}
}

@article{lightman2023verify,
  title   = {Let's Verify Step by Step},
  author  = {Lightman, Hunter and Kosaraju, Vineet and Burda, Yura and Edwards, Harri and Baker, Bowen and Lee, Teddy and Leike, Jan and Schulman, John and Sutskever, Ilya and Cobbe, Karl},
  journal = {arXiv preprint arXiv:2305.20050},
  year    = {2023}
}

@inproceedings{chen2024selfdebug,
  title     = {Teaching Large Language Models to Self-Debug},
  author    = {Chen, Xinyun and Lin, Maxwell and Sch{\"a}rli, Nathanael and Zhou, Denny},
  booktitle = {International Conference on Learning Representations (ICLR)},
  year      = {2024},
  note      = {arXiv:2304.05128}
}

@inproceedings{shinn2023reflexion,
  title     = {Reflexion: Language Agents with Verbal Reinforcement Learning},
  author    = {Shinn, Noah and Cassano, Federico and Gopinath, Ashwin and Narasimhan, Karthik and Yao, Shunyu},
  booktitle = {Advances in Neural Information Processing Systems (NeurIPS)},
  year      = {2023},
  note      = {arXiv:2303.11366}
}

@inproceedings{madaan2023selfrefine,
  title     = {Self-Refine: Iterative Refinement with Self-Feedback},
  author    = {Madaan, Aman and Tandon, Niket and Gupta, Prakhar and Hallinan, Skyler and Gao, Luyu and Wiegreffe, Sarah and Alon, Uri and Dziri, Nouha and Prabhumoye, Shrimai and Yang, Yiming and others},
  booktitle = {Advances in Neural Information Processing Systems (NeurIPS)},
  year      = {2023},
  note      = {arXiv:2303.17651}
}

@inproceedings{gao2023pal,
  title     = {{PAL}: Program-Aided Language Models},
  author    = {Gao, Luyu and Madaan, Aman and Zhou, Shuyan and Alon, Uri and Liu, Pengfei and Yang, Yiming and Callan, Jamie and Neubig, Graham},
  booktitle = {International Conference on Machine Learning (ICML)},
  year      = {2023},
  note      = {arXiv:2211.10435}
}

@inproceedings{pan2023logiclm,
  title     = {{Logic-LM}: Empowering Large Language Models with Symbolic Solvers for Faithful Logical Reasoning},
  author    = {Pan, Liangming and Albalak, Alon and Wang, Xinyi and Wang, William Yang},
  booktitle = {Findings of the Association for Computational Linguistics: EMNLP},
  year      = {2023},
  note      = {arXiv:2305.12295}
}

@inproceedings{ye2023satlm,
  title     = {{SatLM}: Satisfiability-Aided Language Models Using Declarative Prompting},
  author    = {Ye, Xi and Chen, Qiaochu and Dillig, Isil and Durrett, Greg},
  booktitle = {Advances in Neural Information Processing Systems (NeurIPS)},
  year      = {2023},
  note      = {arXiv:2305.09656}
}

@inproceedings{olausson2023linc,
  title     = {{LINC}: A Neurosymbolic Approach for Logical Reasoning by Combining Language Models with First-Order Logic Provers},
  author    = {Olausson, Theo X. and Gu, Alex and Lipkin, Benjamin and Zhang, Cedegao E. and Solar-Lezama, Armando and Tenenbaum, Joshua B. and Levy, Roger},
  booktitle = {Empirical Methods in Natural Language Processing (EMNLP)},
  year      = {2023},
  note      = {arXiv:2310.15164}
}

@article{soeken2018epfl,
  title   = {The {EPFL} Logic Synthesis Libraries},
  author  = {Soeken, Mathias and Riener, Heinz and Haaswijk, Winston and Testa, Eleonora and Schmitt, Bruno and Meuli, Giulia and Mozafari, Fereshte and Micheli, Giovanni De},
  journal = {arXiv preprint arXiv:1805.05121},
  year    = {2018}
}

@inproceedings{schmitt2022tweedledum,
  title     = {tweedledum: A Compiler Companion for Quantum Computing},
  author    = {Schmitt, Bruno and Micheli, Giovanni De},
  booktitle = {Design, Automation \& Test in Europe Conference (DATE)},
  year      = {2022}
}

@inproceedings{wei2022cot,
  title     = {Chain-of-Thought Prompting Elicits Reasoning in Large Language Models},
  author    = {Wei, Jason and Wang, Xuezhi and Schuurmans, Dale and Bosma, Maarten and Ichter, Brian and Xia, Fei and Chi, Ed and Le, Quoc V. and Zhou, Denny},
  booktitle = {Advances in Neural Information Processing Systems (NeurIPS)},
  year      = {2022},
  note      = {arXiv:2201.11903}
}

@inproceedings{wang2023selfconsistency,
  title     = {Self-Consistency Improves Chain of Thought Reasoning in Language Models},
  author    = {Wang, Xuezhi and Wei, Jason and Schuurmans, Dale and Le, Quoc and Chi, Ed and Narang, Sharan and Chowdhery, Aakanksha and Zhou, Denny},
  booktitle = {International Conference on Learning Representations (ICLR)},
  year      = {2023},
  note      = {arXiv:2203.11171}
}

@inproceedings{zhou2023leasttomost,
  title     = {Least-to-Most Prompting Enables Complex Reasoning in Large Language Models},
  author    = {Zhou, Denny and Sch{\"a}rli, Nathanael and Hou, Le and Wei, Jason and Scales, Nathan and Wang, Xuezhi and Schuurmans, Dale and Cui, Claire and Bousquet, Olivier and Le, Quoc and Chi, Ed},
  booktitle = {International Conference on Learning Representations (ICLR)},
  year      = {2023},
  note      = {arXiv:2205.10625}
}

@inproceedings{khot2023decomposed,
  title     = {Decomposed Prompting: A Modular Approach for Solving Complex Tasks},
  author    = {Khot, Tushar and Trivedi, Harsh and Finlayson, Matthew and Fu, Yao and Richardson, Kyle and Clark, Peter and Sabharwal, Ashish},
  booktitle = {International Conference on Learning Representations (ICLR)},
  year      = {2023},
  note      = {arXiv:2210.02406}
}

@article{paltenghi2024survey,
  title   = {A Survey on Testing and Analysis of Quantum Software},
  author  = {Paltenghi, Matteo and Pradel, Michael},
  journal = {arXiv preprint arXiv:2410.00650},
  year    = {2024}
}

@article{burgholzer2020qcec,
  title   = {Advanced Equivalence Checking for Quantum Circuits},
  author  = {Burgholzer, Lukas and Wille, Robert},
  journal = {IEEE Transactions on Computer-Aided Design of Integrated Circuits and Systems},
  volume  = {40},
  number  = {9},
  pages   = {1810--1824},
  year    = {2021},
  note    = {arXiv:2004.08420}
}

@article{peham2022zx,
  title   = {Equivalence Checking of Quantum Circuits with the {ZX}-Calculus},
  author  = {Peham, Tom and Burgholzer, Lukas and Wille, Robert},
  journal = {IEEE Journal on Emerging and Selected Topics in Circuits and Systems},
  volume  = {12},
  number  = {3},
  pages   = {662--675},
  year    = {2022},
  note    = {arXiv:2208.12820}
}

@article{lewis2022formal,
  title   = {Formal Verification of Quantum Programs: Theory, Tools, and Challenges},
  author  = {Lewis, Marco and Soudjani, Sadegh and Zuliani, Paolo},
  journal = {ACM Transactions on Quantum Computing},
  volume  = {5},
  number  = {1},
  pages   = {1--35},
  year    = {2023},
  note    = {arXiv:2110.01320}
}

@article{usandizaga2023mutants,
  title   = {Quantum Circuit Mutants: Empirical Analysis and Recommendations},
  author  = {Mendiluze Usandizaga, E{\~n}aut and Yue, Tao and Arcaini, Paolo and Ali, Shaukat},
  journal = {arXiv preprint arXiv:2311.16913},
  year    = {2023}
}

@article{long2025oracle,
  title   = {A Black-box Testing Framework for Oracle Quantum Programs},
  author  = {Long, Peixun and Zhao, Jianjun},
  journal = {arXiv preprint arXiv:2505.07243},
  year    = {2025}
}

@inproceedings{grover1996,
  title     = {A Fast Quantum Mechanical Algorithm for Database Search},
  author    = {Grover, Lov K.},
  booktitle = {Proceedings of the 28th Annual ACM Symposium on Theory of Computing (STOC)},
  pages     = {212--219},
  year      = {1996}
}

@article{brassard2002amplitude,
  title   = {Quantum Amplitude Amplification and Estimation},
  author  = {Brassard, Gilles and H{\o}yer, Peter and Mosca, Michele and Tapp, Alain},
  journal = {Contemporary Mathematics},
  volume  = {305},
  pages   = {53--74},
  year    = {2002},
  note    = {arXiv:quant-ph/0005055}
}

@article{farhi2014qaoa,
  title   = {A Quantum Approximate Optimization Algorithm},
  author  = {Farhi, Edward and Goldstone, Jeffrey and Gutmann, Sam},
  journal = {arXiv preprint arXiv:1411.4028},
  year    = {2014}
}

@article{che2023trade,
  title={Trade-offs between coherence and mixedness and their evolution under quantum noise channels},
  author={Che, Xu-Jun and Tao, Yuan-Hong and Sheng, Yi-Hao and Wu, Shu-Hui and Fei, Shao-Ming},
  journal={Results in Physics},
  volume={52},
  pages={106794},
  year={2023},
  publisher={Elsevier}
}

\clearpage
\appendix

\section{Reproducibility and Release}
\label{app:release}

The code and data accompanying this paper, released at
\url{https://github.com/chexujun/QEncodeBench}, comprise: (i)~the frozen instance sets (480 core, 180
escalation) with per-instance golden solution sets; (ii)~the
generators with their deterministic seeding scheme: the core set,
fresh regeneration sets, and escalation set draw from three disjoint
seed domains, so any party can mint unlimited fresh instances that
provably do not overlap the frozen ones; (iii)~the verifier with its
16-mutant adversarial suite and the full self-check reports
(references, global-phase variants, cross-method agreement);
(iv)~every run configuration and raw log behind the paper.
Every number in the paper is
recomputable from these logs; verification is
local simulation only.

\section{Verifier Details}
\label{app:verifier}

\paragraph{Why not off-the-shelf verification?}
\emph{Equivalence checkers} (QCEC) decide circuit-vs-circuit
equivalence, so they need a golden reference circuit, but
constructing one is the task under test; moreover their fast path
(random basis-state simulation) is provably weakest exactly on
oracle-style errors, where a faulty multi-controlled gate leaves as
few as $2^{n-c}$ distinguishing
columns~\citep{burgholzer2020qcec}. \emph{ZX-calculus} checking is
incomplete for ancilla-assisted multi-controlled Toffolis, the
standard building block of constraint oracles, and cannot certify
inequivalence or produce counterexamples~\citep{peham2022zx}.
\emph{Formal verification} (quantum Hoare logic, proof assistants)
is interactive, bound to special-purpose languages, and lists
parameterized oracle programs as an open
challenge~\citep{lewis2022formal}. \emph{Statistical testing} of
output distributions is blind to phases ($|+\rangle$ vs
$|-\rangle$) and inherits open threshold/sample-size
problems~\citep{paltenghi2024survey,long2025oracle}; our
interferometric readout is deterministic on noiseless simulation
and needs no thresholds beyond numerical tolerance.

\paragraph{Self-validation, in full.}
(i)~\emph{Adversarial mutants}: 16 wrong-by-design oracles
targeting specific blind spots, namely dropped or flipped clauses,
missing uncomputation (dirty ancillas), off-by-one counter windows,
a mixed-up counter in the hybrid family, and two basis-permutation
mutants crafted to pass the uniform-state screen alone. All 16 are
rejected, and the two permutation mutants are confirmed to fool the
screen without fingerprints; the fingerprints are load-bearing.
Unlike random gate-level mutation used to evaluate \emph{testing
techniques}~\citep{usandizaga2023mutants,burgholzer2020qcec}, these
mutants are semantic and adversarial, and they calibrate the
\emph{verifier}. (ii)~\emph{Method agreement}: the fast path,
exhaustive simulation, and an independent unitary-level check agree
on 500 sampled verdicts (references and mutated variants) with zero
discrepancies. (iii)~\emph{Live spot checks}: on hash-selected real
candidates in every run, the fast path's verdict is cross-checked
against exhaustive simulation; any disagreement freezes the run.
(iv)~\emph{Global-phase regression}: every reference also passes as
its $-1$ and $e^{0.736i}$ global-phase variants.

\paragraph{Mark accuracy.} On L3 failures the verifier reports
\begin{equation}
\alpha(C)=\tfrac{1}{2^n}\,
\max\!\big(2^n-|\Shat\,\triangle\,\Sf|,\;|\Shat\,\triangle\,\Sf|\big),
\label{eq:acc}
\end{equation}
where $\triangle$ is symmetric difference and the $\max$ reflects
the two admissible global-phase hypotheses ($\Shat$ marks $\Sf$ or
its complement).

\paragraph{L1.} The candidate must parse; an AST walk over attribute
chains, plus a recursive scan over instruction names in the built
circuit, rejects oracle-lookup APIs
(\texttt{UnitaryGate}, \texttt{Operator}, \texttt{Statevector},
\texttt{DiagonalGate}, \texttt{Initialize},
\texttt{StatePreparation}, \texttt{PhaseOracle},
\texttt{HamiltonianGate}, \texttt{Isometry}, and their instruction
aliases). The recursive scan closes the loophole of hiding a banned
synthesis inside a custom gate wrapper.

\paragraph{L2.} \texttt{build\_oracle} runs in a subprocess sandbox
(30\,s wall clock); the resulting circuit is rejected if it exceeds
the qubit budget or contains non-unitary instructions
(\texttt{reset}, measurement, \texttt{initialize}), again checked
recursively through custom-gate definitions. The sandbox is a
crash-and-timeout isolation layer, not a security boundary: it
protects the evaluation from pathological candidate code while the
AST screen enforces the task contract; dynamic import, reflection,
and deliberate resource exhaustion sit outside the threat model,
which assumes non-adversarial model outputs run in a disposable
environment.

\paragraph{L3 fast path.} With
$|\psi\rangle = C\,(H^{\otimes n}\!\otimes\!I)|0\rangle$ and
$\psi_x=\langle x,0^a|\psi\rangle$, the checks are
leakage $1-\sum_x|\psi_x|^2\le\varepsilon_\ell=10^{-9}$;
magnitudes $\big||\psi_x|-2^{-n/2}\big|\le\varepsilon_m=10^{-6}$;
phases $\mathrm{dist}\big(\arg(\psi_x/\psi_{0^n}),\,
\pi(f(x){\oplus}f(0^n))\big)\le\varepsilon_\varphi=10^{-6}$.
Each fingerprint probe draws per-amplitude magnitudes from
$U(0.5,1.5)$ and phases from $U(0,2\pi)$ (normalized); the candidate
must act on it exactly as the diagonal read off the uniform state:
output magnitudes equal input magnitudes pointwise
($\varepsilon_m$), output phases consistent with the uniform-state
phase pattern after anchoring both tables at reference coordinate
$0$ ($\varepsilon_\varphi$), leakage within
$\varepsilon_\ell$.

\begin{proposition}[fast-path soundness]
\label{prop:soundness}
Let $C_b$ be the candidate's block matrix
$\langle y,0^a|C|x,0^a\rangle$ and $D$ the diagonal read off the
uniform state. (a)~If $C_b$ is diagonal, fast-path acceptance
implies Definition~\ref{def:correct} up to global phase, with
per-assignment phase error at most $\varepsilon_\varphi$
(deterministic). (b)~If instead
$\min_\gamma\|C_b-e^{i\gamma}D\|_{\max}=\eta$ with
$\eta\ge6\cdot2^{n/2}\varepsilon_m$ (misrouting, leakage, and phase
drift all appear in $\eta$; Appendix~\ref{app:verifier}), then for
every phase
alignment $\gamma$ fixed in advance, the probability over the two
independent fingerprints that both pass the fingerprint clauses at
alignment $\gamma$ (magnitudes preserved pointwise within
$\varepsilon_m$, phases consistent with $e^{i\gamma}D$ within
$\varepsilon_\varphi$) is at most
$\big(3\cdot2^{n/2}\varepsilon_m/\eta\big)^2$: for $n\le12$, a
misrouted basis state ($\eta\approx1$) escapes with probability at
most $4\times10^{-8}$.

\end{proposition}

\begin{proof}[Proof of Proposition~\ref{prop:soundness}]
\emph{Part (a).} If $C_b=\mathrm{diag}(e^{i\theta_x})$ with clean
ancillas, the uniform-state amplitudes are exactly
$2^{-n/2}e^{i\theta_x}$. The leakage and magnitude conditions hold
identically, and the phase condition checks
\[
\mathrm{dist}\big(\theta_x-\theta_{0^n},\;
\pi(f(x){\oplus}f(0^n))\big)\;\le\;\varepsilon_\varphi
\]
pointwise, which is Definition~\ref{def:correct} up to global phase
with per-assignment error at most $\varepsilon_\varphi$. No
randomness is involved.

\smallskip
\emph{Part (b): setup.} Fix the alignment $\gamma$ from the
statement and write $M=C_b-e^{i\gamma}D$; since $\eta$ is the
minimum over alignments, $\|M\|_{\max}\ge\eta$, so we may pick an
entry with $|M_{rs}|\ge\eta$. A
fingerprint is
\[
c_y=\frac{m_y}{N}\,e^{i\phi_y},\qquad
m_y\sim U(0.5,1.5),\quad \phi_y\sim U(0,2\pi)\ \text{i.i.d.},
\]
with $N^2=\sum_y m_y^2$, so that
$0.5\cdot2^{n/2}\le N\le 1.5\cdot2^{n/2}$ and $|c_y|\le1.5/N$ for
every $y$.

\smallskip
\emph{What passing at alignment $\gamma$ implies.} The magnitude
clause bounds $\big||\mathrm{out}_x|-|c_x|\big|\le\varepsilon_m$,
and the phase clause at alignment $\gamma$ bounds the pointwise
phase error relative to $e^{i\gamma}$ by $\varepsilon_\varphi$.
Together, for every $x$,
\[
\big|\mathrm{out}_x-e^{i\gamma}e^{i\theta_x}c_x\big|
\;\le\;\varepsilon_m+|c_x|\,\varepsilon_\varphi
\;\le\;t:=\varepsilon_m+\tfrac{1.5}{N}\,\varepsilon_\varphi .
\]
Because the same $\gamma$ appears in $M$, the residual at row $r$
obeys $|(Mc)_r|\le t$.  (This step is where fixing the alignment is
load-bearing: with the implementation's probe-anchored reference
phase in place of $\gamma$, the decomposition below would lose the
independence of $A$ and $R$ from $\phi_s$, and the small-ball step
would not apply; see the scope note after the proof.)

\smallskip
\emph{Small-ball step.} Isolate the contribution of the single
random phase $\phi_s$:
\[
(Mc)_r=A\,e^{i\phi_s}+R,\qquad
A=\frac{M_{rs}\,m_s}{N},\qquad
|A|\;\ge\;\frac{0.5\,\eta}{N},
\]
where $R$ collects the remaining terms and is independent of
$\phi_s$. Under the hypothesis $\eta\ge6\cdot2^{n/2}\varepsilon_m$
we have $t/|A|\le 2tN/\eta<1$ \emph{strictly} (see the constants
step), so the arc geometry below applies. Strictness matters: at
the degenerate corner $R=0$, $t=|A|$ the event is the whole circle
and the chord bound fails; with $t<|A|$ the case $R=0$ has empty
event, and for $R\ne0$ the event $|Ae^{i\phi_s}+R|\le t$ confines
$e^{i\phi_s}$ to an arc whose endpoints lie in a disk of radius
$\rho:=t/|A|$: the arc has half-angle $\theta^*$ with
$\cos\theta^*\ge0$ (from $\rho\le1$) and $\sin\theta^*\le\rho$
(endpoints within chord $2\rho$), hence measure
$2\theta^*\le2\arcsin\rho\le\pi\rho$, so
\[
\Pr\big[\,|(Mc)_r|\le t\,\big]
\;\le\;\frac{\arcsin(t/|A|)}{\pi}
\;\le\;\frac{t}{2|A|}
\;\le\;\frac{tN}{\eta}.
\]

\smallskip
\emph{Constants.} With $N\le1.5\cdot2^{n/2}$ and
$\varepsilon_\varphi=\varepsilon_m$,
\[
tN=\varepsilon_m N+1.5\,\varepsilon_\varphi
\;\le\;1.5\cdot2^{n/2}\varepsilon_m+1.5\,\varepsilon_m
\;<\;3\cdot2^{n/2}\varepsilon_m ,
\]
the last inequality strict for $n\ge1$ (this is what makes
$t/|A|<1$ strict above, with room to spare: in fact
$t/|A|\le\tfrac12+2^{-n/2-1}$). A single fingerprint therefore
passes at alignment $\gamma$ with probability at most
$3\cdot2^{n/2}\varepsilon_m/\eta$. The two fingerprints are drawn
independently, which gives the square.

\smallskip
\emph{Leakage reduction.} If basis state $s$ leaks mass $L$ out of
the computational block, the block column satisfies
$\|C_b e_s\|^2=1-L$, hence $|(C_b)_{ss}|\le\sqrt{1-L}$ and
\[
|M_{ss}|\;\ge\;1-\sqrt{1-L}\;\ge\;L/2 ,
\]
so leakage enters the bound through $\eta\ge L/2$.
\end{proof}

\noindent\emph{Scope of the guarantee.} Two freedoms are outside
the bound. \emph{Alignment.} The proposition bounds acceptance at
each phase alignment fixed in advance, whereas the implemented
check anchors the probe's phase table at a reference coordinate
(coordinate $0$), so the alignment a probe is judged against is a
function of the probe itself. The reduction from the anchored event
to a fixed alignment does not follow from the arguments above (the
alignment ranges over a continuum, so a union bound is unavailable,
and a probe-dependent alignment breaks the independence used in the
small-ball step); we therefore state and prove the per-alignment
form. A candidate could exploit the anchoring only by
presenting a different near-diagonal match to different probes;
the adversarial mutant suite shows no such behavior, and every
acceptance is confirmed exhaustively regardless. \emph{Seed.} The
bound treats the candidate as fixed and the fingerprints as random;
operationally the fingerprint seed is fixed and published for
reproducibility, so the guarantee covers errors made independently
of that seed, which is the case for model-generated circuits. A hypothetical adversary who crafts a circuit against the published
seed is outside the bound, but cannot reach a verdict either way:
the exhaustive confirmation of every acceptance is seed-free.
The two permutation mutants (M13, M14) empirically confirm both
halves of the proposition: they pass the uniform-state screen alone
and are rejected once the fingerprints run.

\paragraph{Exhaustive confirmation of all acceptances.} Every L3 acceptance reported in this paper carries an exhaustive
per-basis-state verdict: all accepted samples, across the in-house
runs and the delivered rows (the open-weight table, both Claude
rows, and the GPT-5.6 Terra row), were re-run with the exhaustive
stage forced, and no fast-path acceptance was overturned. The confirmation harness and per-sample logs are included in the
released code and data.

\paragraph{Escalation and authority.} The fast path escalates to
exhaustive per-basis-state simulation when any edge signal fires:
phase deviation $>\varepsilon_\varphi/10$, leakage
$>\min(10^{-10},\,\varepsilon_\ell/(10\cdot 2^n))$ (the leak edge
scales with dimension because exhaustive bounds leakage per state),
or fingerprint magnitude deviation $>\varepsilon_m/10$, subject to
an in-process cost gate $2^n\cdot\mathrm{size}(C)\le 5\times10^8$.
Exhaustive verdicts are authoritative in both directions.
Independently, a hash-selected 5\% of all real candidates are
cross-checked exhaustively regardless of edge signals; any
disagreement freezes the run. Failure diagnosis (marked set, mark
accuracy, counterexamples) runs exhaustively for $n\le10$ and
circuits up to 30{,}000 instructions, reporting up to 8
counterexamples.

\paragraph{Isolation.} Circuits containing instructions wider than
six qubits, or otherwise flagged as simulation risks, run L3 in an
isolated subprocess (180\,s); transpilation for the L4 depth check
is likewise isolated (120\,s). Timeouts score as failures with a
distinct reason code, never as passes.

\paragraph{The 16 adversarial mutants.} Each targets a specific
blind spot; all 16 are rejected by the full verifier.
M01 dropped constraint; M02 partial phase flip; M03 dirty ancilla
(missing uncomputation); M04 flipped literal polarity; M05
off-by-one qubit index; M06 missing phase application; M07 AND for
OR aggregation; M08 wrong phase pattern; M09 doubled phase
(self-cancelling); M10 non-surjective color comparator; M11 wrong
target value; M12 wrong cardinality bound; M13 basis swap; M14
phased 3-cycle tuned to cancel on the uniform state; M15 rolling
counter window off by one; M16 hybrid-family counter mix-up
(cardinality check wired to the clause counter).

\paragraph{Simulation hygiene.} The exhaustive path injects each
basis state directly into the simulator; the injected state is kept
outside the transpiler (only the payload circuit is transpiled,
under a fixed transpiler seed), because optimizer passes may
otherwise commute phase gates across the state injection and
corrupt per-basis phases. A regression test pins this requirement.

\section{Budget Formula Derivations}
\label{app:budgets}

Budgets equal the qubit footprint of the tier's reference
construction; $\wcw(x)=\lceil\log_2(x{+}1)\rceil$ counts counter
bits. \emph{Flag tiers}: one flag per constraint plus one slack
ancilla for models that prefer computing the aggregate into a final
AND ancilla: 3-SAT $n+m+1$; coloring $2V+E+1$ (two qubits per
vertex under the surjective color code); Latin squares $2F+P+1$
($F$ free cells, $P$ constraint pairs); string matching
$n+\mathit{offs}+1$ (one flag per pattern offset), or $n+2$
anchored. \emph{Counter tiers}: one shared constraint flag plus a
$\wcw(\cdot)$-bit counter, phase on ``count equals all'': 3-SAT
$n+1+\wcw(m)$; coloring $2V+\wcw(E)+1$; Latin squares
$2F+\wcw(P)+2$; string matching $n+1+\wcw(\mathit{offs})$ (phase on
count $\ge$ 1). \emph{Registers that cannot be compressed}: subset
sum keeps the raw sum register, $n+\wcw(\sum_i v_i)$ at every tier;
vertex cover needs edge flags and a weight counter, $n+E+\wcw(n)$;
the hybrid family needs clause flags and a popcount counter,
$n+m+\wcw(n)$. \emph{Zero-slack tier}: the exact footprint of a
rolling-computation reference that recomputes each constraint into
one shared flag; vertex cover drops the per-edge flags for an edge
counter, $n+1+\wcw(E)+\wcw(n)$, the tightest cell in the benchmark.

Two candidate escalation cells were rejected by measurement, per
the pre-registered yield rule. Coloring at scale ($V{=}6$,
$E\ge10$): across 30 sampled instances the maximum solution density
was $0.0137 < 1/64$; the family already sits at the density floor
at its top graded tier. Latin squares on a $4\times4$ board with an
exact (non-surjective) code: across 3{,}000 sampled specifications
the maximum solution count was 2, far below the $\rho\ge1/64$
floor. Accepted escalation cells had healthy yield: the worst
rejection rate during freezing was 5\% (string matching,
zero-slack); all other cells rejected 0\%.

\section{Prompts and Contracts}
\label{app:prompts}

\paragraph{Instance prompt.} Static material first (role, rules,
the interface contract below, one worked toy example), instance
material last (problem text with the formal specification, the
qubit budget, and the depth cap), an ordering chosen so the long
static prefix is cacheable across the whole benchmark. The terse
variant compresses the problem rendering only; the contract is
identical (robustness: 10.5\% vs 10.2\% semantic pass@1, no-think).
The contract requires
\texttt{build\_oracle(qc, problem\_qubits, ancilla\_qubits)} to
append gates only, realize the $(-1)^{f(x)}$ phase semantics of
Definition~\ref{def:correct}, restore all ancillas to
$|0\rangle$, and avoid the
banned high-level APIs of Appendix~\ref{app:verifier}.

\paragraph{Unit-agent contract and bench.} The model receives the
constraint list and writes
\texttt{build\_constraint\_j(qc, problem\_qubits, flag, scratch)}
per predicate, contracted to Eq.~\eqref{eq:fragment}: flip
\texttt{flag} iff
$p_j(x)$, no phase side effects, scratch returned to
$|0\rangle$. Unit checks run each fragment in isolation on a bench
whose wiring deliberately differs from deployment (flag on the last
wire, scratch first), so fragments that hardcode wire indices fail
their unit checks rather than the assembled circuit; the checks are
predicate agreement on all $2^n$ inputs, phase cleanliness, scratch
restoration, and a fingerprint identity test. Failing fragments are
repaired in at most two batched rounds with per-fragment
counterexamples. Assembly is deterministic: compute all flags,
phase on the aggregate, uncompute in mirror order. The plan is chosen by budget: per-flag when $m+s$ ancillas
fit; rolling (one shared
flag, a $\wcw(m)$ counter) otherwise; hold-out rolling (roll
$m{-}1$ fragments, keep the last computed through the phase) for
budgets where counting to $m$ needs one more bit than counting to
$m{-}1$.

\paragraph{NSC contract.} The model returns one JSON term over the
seven IR primitives. Two worked examples ship in the prompt; e.g.\
a 3-SAT clause set becomes
{\footnotesize\begin{verbatim}
{"op":"and","args":[
 {"op":"or","args":[
   {"op":"bit","i":0,"val":1},
   {"op":"bit","i":3,"val":0},
   {"op":"bit","i":4,"val":1}]},
 ...]}
\end{verbatim}}
\noindent and a subset-sum instance becomes
{\footnotesize\begin{verbatim}
{"op":"linsum_cmp",
 "terms":[[5,0],[3,1],[7,2]],
 "cmp":"==","rhs":10}
\end{verbatim}}
\noindent The interpreter validates the term against the instance's
solution set on all $2^n$ assignments, returning concrete
counterexamples for at most two retry rounds. The compiler selects
among five plans: per-child flags with one aggregate phase; a raw
linsum register with phase enumeration over matching values; a
rolling counter (AND roots phase on ``count equals all,'' OR roots
on every nonzero count); and rolling combined with a linsum
register. Contradictory or constant subterms are folded before
planning; every compiled circuit still passes the full verifier.

\section{Agent-Harness Protocols}
\label{app:harness}

Three rows of Table~\ref{tab:main} were produced through agent
harnesses rather than the raw model API. They are indicative
reference points: the harness system prompt differs from a raw API
call, so they are not strictly protocol-comparable to the API
rows.

\paragraph{Claude rows ($^{\dagger}$).} Claude Opus 4.8 and Claude
Haiku 4.5 were invoked through the Claude Code agent harness, one
attempt per instance, with extended thinking and tool use
disabled.

\paragraph{GPT-5.6 Terra row ($^{\ddagger}$).} GPT-5.6 was invoked
at the Terra tier of the Codex agent, one fresh subagent per
instance, answering in one shot without tool use.
Every solver message is the frozen verbose prompt itself: each of
the 480 payloads is rendered by \texttt{build\_prompt}, archived,
and SHA-256-hashed per instance, and we re-derived all 480
renderings independently with byte-identical results. All responses were re-verified in our environment with every
acceptance confirmed exhaustively; three collaborator-side
verification timeouts pass under our longer compute budget and are
scored as passes. The per-tier profile remains non-monotone
(82.3/88.0/79.1).

\paragraph{Per-family, per-tier direct generation (semantic
$\passk{1}$, \%).} Graded tiers 1/2/3; reasoning off $\to$ on:
3-SAT 8.2/9.7/6.1 $\to$ 58.8/57.6/36.4;
coloring 1.8/5.5/0.0 $\to$ 41.2/33.3/0.0;
vertex cover 0.0/0.9/2.6 $\to$ 20.8/13.0/4.3;
subset sum 15.8/11.3/1.7 $\to$ 91.7/60.9/34.8;
Latin squares 38.3/13.0/5.2 $\to$ 91.7/69.6/34.8;
string matching 34.2/25.2/14.8 $\to$ 58.3/78.3/47.8.
Coloring at tier 3 is unsolved in both modes (0/33 tasks), and
vertex cover never exceeds 21\%: the two families that require
dedicated encoding techniques dominate the difficulty at every
tier.

\paragraph{Open-weight serving and token budgets.} All open-weight
models were served with vLLM (OpenAI-compatible endpoint);
Qwen2.5-Coder-32B, Qwen3-32B, and Qwen3-235B-A22B-Thinking in
AWQ int4,
DeepSeek-R1-0528-Qwen3-8B in bf16, gpt-oss-20b/120b in MXFP4.
Prompts are byte-identical to the in-house runs (SHA-256 verified
per instance). Output budgets are sized so that no model is limited
by its chain-of-thought length: 3k tokens for non-reasoning models,
32k for Qwen3/R1 reasoning models (their chains fit well inside it),
and up to 90k for high-effort gpt-oss. A response still reasoning
at its cap yields no code and scores as a failure ($\sim$7\% of
gpt-oss-20b instances; no other model is affected). The gpt-oss
L3$\to$L4 gaps (8.4 points / 12.2 points) are exclusively
\texttt{DEPTH\_EXCEEDED}: transpiled depth at $4.5$--$10\times$ the
reference oracle (median $5.5\times$) against the $4\times$
budget.

\section{Additional Results}
\label{app:results}

\paragraph{Degraded-scoring detail.} Every direct-generation sample
behind Table~\ref{tab:main} is re-scored with its L3 verdict as
ground truth; basis-state checks draw $k\in\{16,64\}$ seeded per
task, probes compare four random product states against the ideal
output state. Table~\ref{tab:scoring} gives the per-row rates.

\begin{table}[t]
\centering
\small
\setlength{\tabcolsep}{4.5pt}
\begin{tabular}{@{}llccc@{}}
\toprule
Model / mode & Tests & Scored & True & FP \\
\midrule
\multirow{2}{*}{Qwen2.5-Coder-7B} & basis-64 & 27.7 & 0.0 & 100\% \\
 & probes-4 & 0.0 & 0.0 & --- \\
\addlinespace
\multirow{2}{*}{DeepSeek-V4-Flash} & basis-64 & 35.1 & 10.2 & 71.0\% \\
 & probes-4 & 10.3 & 10.2 & 1.2\% \\
\addlinespace
\multirow{2}{*}{\;\;+ thinking} & basis-64 & 66.2 & 45.4 & 31.4\% \\
 & probes-4 & 45.4 & 45.4 & 0.0\% \\
\bottomrule
\end{tabular}
\caption{Degraded-scoring study on the core set. ``Scored'' is
$\passk{1}$ under the degraded regime, ``True'' the L3 rate, ``FP''
the fraction of accepted circuits that are semantically wrong
(AlphaCode convention;~\citealp{li2022alphacode}).}
\label{tab:scoring}
\end{table}

\paragraph{Depth-cap sensitivity.} Holding the qubit gate and
transpiler settings fixed and sweeping the depth cap over
$\{2,4,8\}\times d_{\mathrm{ref}}$ (the published column is
$4\times$), full pass rates are: direct, reasoning off (480 set)
8.5/9.3/10.0; direct, reasoning on (480) 38.8/40.8/44.0; and on the
126-task subset, direct thinking 43.1/45.0/47.4, unit agent
63.0/80.4/89.7, NSC 90.7/91.8/92.1. The NSC-over-unit ordering
holds at every multiplier; the gap is 27.7 points at $2\times$ and 2.4 points
at $8\times$.

\paragraph{Depth premium by architecture.} The semantic-to-full gap
(L4 depth gate) is architecture-dependent (126-task subset): direct
3.9 points, few-shot 3.4 points, decomposition 0.8 points, self-repair 0 points, NSC
0.3 points. The unit
agent, by contrast, pays 12.7 points, and its gap \emph{grows with
repair dose} (54-task ablation subset: 8.0 points at 0
rounds, 9.2 at 1, 13.6 at 2; 11.7 under which-only feedback): each
repair round tends to add gates, so unit-level repair buys
semantic correctness partly at the price of depth. NSC's compiled
circuits stay well inside the cap.

\paragraph{Sampling curves.} Direct, reasoning on:
$\passk{1/2/3} = 48.9/66.4/73.0$. Direct, reasoning off:
$\passk{1/3/5} = 10.8/23.7/32.5$ (126-task subset).

\paragraph{Escalation, full gate.} L4-gated rates track the
semantic rates of Table~\ref{tab:escalation} within two points throughout:
$\passk{3}$(L4) is 81.5\% (unit agent), 100\% (NSC), 24.1\%
(direct).

\paragraph{Failure gallery.} The verifier's single-edit diagnosis
(matching $\Shat(C)$ against libraries of plausible mis-encodings:
dropped constraint, flipped literal polarity, off-by-one bounds,
non-surjective color decoding) exactly explains a substantial
fraction of L3 failures, and error patterns are consistent across
instances, echoing~\citet{yang2025qcircuitbench}. Two
representative L3 failures from the reasoning-on direct run. (i)~\texttt{3sat-T1-s0008},
\textsc{mark-mismatch} with mark accuracy $0.875 = 7/8$: exactly
the signature of one ignored three-literal clause (a single clause
excludes $1/8$ of assignments); counterexamples 5 and 10 satisfy
the other clauses but violate the dropped one. (ii)~%
\texttt{3sat-T3-s0029}, \textsc{ancilla-dirty}: the candidate
computes clause flags correctly but skips uncomputation; the
residual entanglement destroys the interference pattern even
though every clause flag is individually correct: the failure
mode that basis-state testing cannot see and that motivates mutant
M03.

\end{document}